\documentclass[12pt,reqno,fleqn,a4paper]{amsart}
\usepackage{mathrsfs}
\usepackage{bbm}
\usepackage{amsmath,amssymb}
\usepackage{color}

\usepackage[titletoc]{appendix}

\allowdisplaybreaks \numberwithin{equation}{section}
\numberwithin{equation}{section} 

 \newtheorem{thm}{Theorem}[section]
 
 \newtheorem{lem}[thm]{Lemma}
 \newtheorem{prp}[thm]{Proposition}

\newenvironment{prf}{\noindent {\it Proof} \ }{\hfill $\Box$}

\newcommand{\eqa}{\begin{eqnarray}}
\newcommand{\eeqa}{\end{eqnarray}}

\newcommand{\beq}{\begin{equation}}
\newcommand{\eeq}{\end{equation}}

\renewcommand{\[}{\begin{eqnarray}}
\renewcommand{\]}{\end{eqnarray}}
\newcommand\pd{\partial}

\newcommand{\nn}{\nonumber}

\newcommand{\Gm}{\Gamma}

 \renewcommand{\d}{\partial}

\newcommand{\res}{\mathrm{res}}

\newcommand\Q{\mathbb{Q}}

\newcommand\Z{\mathbb{Z}}

\newcommand\cA{\mathcal{A}}
\newcommand\B{\mathcal{B}}

\newcommand\cL{\mathcal{L}}
\newcommand\cM{\mathcal{M}}

\begin{document}
\title{Supersymmetric BKP systems and their symmetries}
\author{
Chuanzhong Li\dag,\  \ Jingsong He\ddag}
\allowdisplaybreaks
\dedicatory {\small Department of Mathematics and Ningbo Collabrative Innovation Center of Nonlinear Harzard System of Ocean and Atmosphere,\\
 Ningbo university, Ningbo 315211, China,\\
\dag Email:lichuanzhong@nbu.edu.cn\\
\ddag Email:hejingsong@nbu.edu.cn}
\thanks{\ddag Corresponding author}

\date{}

\begin{abstract}
In this paper, we construct the additional symmetries of the supersymmetric BKP(SBKP) hierarchy.
These additional  flows constitute a B type $SW_{1+\infty}$ Lie algebra because of the B type reduction of the supersymmetric BKP  hierarchy. Further we generalize the  SBKP hierarchy to a supersymmetric two-component BKP (S2BKP) hierarchy
equipped with a B type $SW_{1+\infty}\bigoplus SW_{1+\infty}$ Lie algebra. As a Bosonic reduction of the  S2BKP hierarchy, we define a new constrained system called the supersymmetric Drinfeld-Sokolov hierarchy of type D which admits a $N=2$ supersymmetric Block type symmetry.
\end{abstract}

\allowdisplaybreaks
\maketitle{\allowdisplaybreaks}

\tableofcontents
\allowdisplaybreaks
 \setcounter{section}{0}

\section{Introduction}

In the study of  integrable hierarchies, it  is interesting
to find their symmetries and identify
the algebraic structure of the symmetries.  Among
these symmetries, the additional symmetry is an important type
 which contains
dynamic variables explicitly and these additional flows do not commutes with each other.
Additional symmetries of the Kadomtsev-Petviashvili(KP) hierarchy
were introduced  by Orlov and Shulman \cite{os1} which contain one important symmetry, i.e. the so-called Virasoro symmetry.  These
symmetries form a centerless $W_{1+\infty}$ algebra closely
related to matrix models by means of the Virasoro constraint and
string equations\cite{D witten,Douglas, dl1,asv2}. Two
sub-hierarchies as the BKP hierarchy and CKP
hierarchy\cite{DKJM-KPBKP,DJKM,takasaki,tu,heLMP,jipeng}  have been shown to possess
additional symmetries,
 with consideration of
the reductions on the Lax operators.

Various generalizations and supersymmetric extensions \cite{yamada} of the KP hierarchy have
deep implications in mathematical physics, particularly in the theory of Lie algebras. In \cite{kacBF,kacBFB}, the  theory of the super Lie algebras was surveyed by considering super Boson-Fermion Correspondences. One important supersymmetric extension is the supersymmetric Manin-Radul Kadomtsev-Petviashvili
(MR-SKP) hierarchy\cite{maninsuperKP} which contains a lot of integrable super solitary equations equipped with the super-pseudodifferential
operators. Apart from the Manin-Radul one, Mulase supersymmetrize
the KP hierarchy  by constructing a hierarchy  called the
Jacobian SKP hierarchy
which does not possess a standard Lax formulation \cite{Mulase}. This hierarchy has strict Jacobian
flows, i.e. it preserves the super Riemman surface about which one can also see \cite{rosly}.
The additional symmetries for super hierarchies  were firstly found in the paper \cite{orlovsuper} by constructing the standard Orlov-Schulman additional
nonisospectral  flows. Later the additional symmetry of the MR-SKP hierarchy was studied by Stanciu \cite{Stanciu}. The
 ghost symmetries, hamiltonian structures and extensions of the MR-SKP hierarchy were studied as well as reductions of the MR-SKP hierarchy\cite{aratynsuper,superKPhe}.  Later the supersymmetric BKP (SBKP) hierarchy  was constructed in \cite{StanciuBKP}. After that this series of super hierarchies  were seldom studied in mathematical physics partly because of their extreme complexities.

 For the symmetry of the two-component BKP hierarchy, there is a series of works such as \cite{DJKM,shiota,LWZ,bkpds}.
 In the paper \cite{BKP-DS}, we construct the generalized additional symmetries of the two-component BKP hierarchy
 and identify its algebraic structure.
 Besides, the D type Drinfeld-Sokolov hierarchy was found to be a good differential model to derive a complete
 Block type infinite dimensional Lie algebra. About the Block algebra related to integrable systems, we did a series of works in \cite{ourBlock}-\cite{torus}.
 In this paper, we will construct the additional symmetries of the supersymmetric BKP hierarchy.
These additional flows constitute a B type $SW_{1+\infty}$ Lie algebra. Further we generalize the SBKP hierarchy to a supersymmetric two-component BKP hierarchy (S2BKP) hierarchy
and derive its algebraic structure. As a reduction of the S2BKP hierarchy, in this paper a new supersymmetric Drinfeld-Sokolov hierarchy of type D will be constructed and proved to have a super Block type additional symmetry.

In (1+1) dimensional sypersymmetric integrable systems, starting from 1980s, there is also a series of work about superversions of Korteweg-de Vries(KdV), Toda-KdV equations and so on\cite{kulishkdv}-\cite{toda-kdv}. Most of them are related to the conformal field theory and string theory\cite{string,KDV}. The generalized KdV equations
and Toda lattice equations are particularly interesting integrable nonlinear systems in connection with
conformal field theories. Their Virasoro symmetry can be extended to a $W_n$ algebra
which is known to
arise from the Hamiltonian structure of the generalized KdV equation \cite{DS} by incorporating conserved currents of higher spins.
The supersymmetric version of the Drinfeld-Sokolov reduction  of the Toda-KdV theories gives the generators of the related super $W$-algebra
with the commutation relations provided by the associated Hamiltonian structure \cite{kannosuperDS}.
In this paper we will extend the Lie algebraic method of Drinfeld and Sokolov \cite{DS} to
the sypersymmetric case and develop a Lie superalgebraic method for a supersymmetric D type Drinfeld-Sokolov hierarchy. We further derive that the supersymmetric Drinfeld-Sokolov hierarchy of type D  possesses a $N=2$ supersymmetric Block type Lie algebraic structure.

This paper is arranged as follows. In the next section we recall some
necessary facts of the SBKP hierarchy. In Sections 3,
we will give the additional symmetries for the
SBKP hierarchy. The ghost symmetry  of the SBKP hierarchy will be devoted to Section 4 and 5 using the techniques in \cite{aratynsuper}.  Further in Section 6 and 7, we generalize the SBKP hierarchy to a S2BKP  hierarchy
and derive its B type $SW_{1+\infty}\bigoplus SW_{1+\infty}$  algebra.
 As a Bosonic reduction of the  supersymmetric two-component BKP hierarchy, we define a new constrained system called the supersymmetric Drinfeld-Sokolov hierarchy of type D which possesses a $N=2$ supersymmetric Block type Lie algebra in the following two sections.
 Finally, we will give a short conclusion and a further discussion.

\section{The supersymmetric BKP hierarchy}

Let us firstly recall some basic facts\cite{StanciuBKP} on the supersymmetric  BKP  system which
is well defined by two Lax operators.

 $\cA$ is assumed as an algebra of smooth functions of a spatial
coordinate $x$, a grassmann variable $\theta$ and their super-derivation denoted as  $ D=\d_{\theta}+\theta\d$. This algebra  $\cA$ has the following multiplying rule

\begin{align}\label{Apm}
D^n\Phi &= \sum_{i=0}^\infty {n\brack {n-i}} (-1)^{|\Phi|(n-i)} \Phi^{[i]}
D^{n-i},\end{align}
\eqa{n\brack{n-i}} = \begin{cases}
0& i<0\ \  or\ \  (n,i)= (0,1)\pmod{2};\\
{\left({{\left[\frac n2\right]}\atop{\left[\frac{ n-i}2\right]}}\right)}&
i\geq 0,(n,i)\neq(0,1)\pmod{2}.\end{cases}
\eeqa

Here the value $|\Phi|$ means the super degree of the operator $\Phi$ which shows  the operator $\Phi$ is Fermionic or Bosonic.
The supersymmetric derivative $D$ satisfies the supersymmetric analog of the Leibniz rule
\eqa D(ab)=D(a)b+(-1)^{|a|}a D(b),\eeqa
where $a$ is a homogeneous element of $A$.
We introduce the even and odd time variables $(t_2,t_3,t_6,t_7,\cdot)$ and the following definition of even and odd flows
\eqa
D_{4i-2}=\frac{\d}{\d t_{4i-2}},\ \ D_{4i-1}=\frac{\d}{\d t_{4i-1}}+\sum_{j=1}^{\infty}t_{4j-1}\frac{\d}{\d t_{4i+4j-2}}.
\eeqa
We recall that the supercommutator is defined as $[X,Y]=XY-(-1)^{|X||Y|}YX.$
The bracket has a property as $[X,YZ]=[X,Y]Z+(-1)^{|X||Y|}Y[X,Z].$
Then $D^2=\frac12[D,D]=\d.$
This  family of
infinite odd and even flows satisfy a nonabelian Lie superalgebra whose
commutation relations are
$$
[D_{4i-2},D_{4j-2}]=0~,\quad [D_{4i-2},D_{4j-1}]=0~,\quad
[D_{4i-1},D_{4j-1}]=-2 D_{4i+4j-2}~,
$$
\eqa[D_{4i-2},D]=0~,\quad [D_{4i-1},D]=0.\eeqa
For any operator $A=\sum_{i\in\Z} f_i D^i\in\cA$ and homogeneous operators $P,Q$, its
nonnegative projection, negative projection, adjoint operator are
respectively defined as
\begin{align}\label{Apm}
&A_+=\sum_{i\geq0} f_i  D^i, \quad A_-=\sum_{i<0} f_i  D^i, \quad A^*=\sum_{i\in\Z}(- D)^i\cdot f_i,\\
& (PQ)^*=(-1)^{|P||Q|}Q^*P^*,\ \ (P^{-1})^*=(-1)^{|P|}(P^{*})^{-1}.
\end{align}
Also for the operator $D^k$, the adjoint operator is defined as
\begin{align}
(D^k)^*=(-1)^{\frac{k(k+1)}{2}}D^k.
\end{align}

Basing on definitions  in \cite{StanciuBKP}, the Lax operator of the supersymmetric  BKP hierarchy has a form as
\begin{equation} \label{PhP}
L= D+\sum_{i\ge1}u_i  D^{1-i},\ \ \ u_2=-\frac12u_1^{[1]}.\end{equation}
The supersymmetric  BKP hierarchy is defined by the following
Lax equations
\begin{align}\label{bkpLax}
& D_{4k-2} L=[(L^{4k-2})_+,  L], \quad  \quad  D_{4k-1}  L=[(L^{4k-1})_+,  L]-2L^{4k}, \ \ k\geq 1.
\end{align}

One can rewrite the operator $L$ in a dressing form as
\begin{equation} \label{PPh}
L=\Phi D\Phi^{-1},
\end{equation}
where
\begin{align} \label{dreop}
\Phi=1+\sum_{i\ge 1}a_i D^{-i},
\end{align}
 satisfy
\begin{equation}\label{phipsi}
\Phi^*= D\Phi^{-1} D^{-1}.
\end{equation}

We call the eq.\eqref{phipsi} the B type condition of the supersymmetric  BKP hierarchy.
Given $L$, the dressing operator $\Phi$  is determined uniquely up to a multiplication to the
right by operators with
constant coefficients. The dressing operator $\Phi$ takes values in a B type Volterra group. The supersymmetric  BKP hierarchy \eqref{bkpLax}
can also be redefined as
\begin{align}\label{satoeq}
&\frac{\pd \Phi}{\pd t_{4k-2}}=- (L^{4k-2})_-\Phi, \quad \frac{\pd \Phi}{\pd t_{4k-1}}=- (L^{4k-1})_-\Phi,
\end{align}
with $k\geq 1$.

With the above preparation, it is time to  construct additional symmetries for the supersymmetric  BKP hierarchy in the next section.

\section{Additional symmetries of the supersymmetric BKP hierarchy}

In this section, we are to construct  additional symmetries for the supersymmetric BKP hierarchy by using the Orlov--Schulman operators whose coefficients
depend explicitly on the time variables of the
hierarchy. The Orlov--Schulman operators $M_i$ and auxiliary operator $\Q$ are constructed in the following
dressing structure
\begin{equation*}\label{}
M_i=\Phi\Gm_i\Phi^{-1}, \quad\ i=0,1;\ \ \Q=\Phi Q\Phi^{-1},
\end{equation*}
where
\eqa \notag
\Gm_0 &=&x+\frac12\sum_{k\geq 1}(4k-2)t_{4k-2} D^{4k-4}+\frac12(4k-1)t_{4k-1} D^{4k-3}\\
&&-\frac12\sum_{k\geq 1}t_{4k-1} \d^{2k-2}Q+\sum_{i,j\geq 1}(i-j)t_{4i-1}t_{4j-1} \d^{2i+2j-2},\\
\Gm_1 &=&\theta+\sum_{k\geq 1}t_{4k-1} \d^{2k-1},
\eeqa
where $Q=\d_{\theta}-\theta\d.$

Then one can get the following lemma.

\begin{lem}
The operators $\Gamma_j, Q$ satisfy
\[
[D_{4i-2}-D^{4i-2}, \Gamma_j]&=&[D_{4i-1}-D^{4i-1}, \Gamma_j]=0;\ \ j=0,1,\]

\[
[D_{4i-2}-D^{4i-2}, Q]&=&[D_{4i-1}-D^{4i-1}, Q]=0,
\]

\[
[Q, \Gamma_0]&=&-\Gamma_1,\ \ \ [Q, \Gamma_1]=1,\ \ \ [\d, \Gamma_0]=1.
\]

\end{lem}

\begin{proof}
For the proof, one can do the following direct calculation
\[
[D_{4i-2}-D^{4i-2}, \Gamma_0]&=&\frac12(4i-2) D^{4i-4}-[D^{4i-2}, x]\\
&=&(2i-1) \d^{2i-2}-[\d^{2i-1}, x]=0,
\]

\[&& \notag
[D_{4i-1}-D^{4i-1}, \Gamma_0]\\
&=&[\frac{\d}{\d t_{4i-1}}-\sum_{j=1}^{\infty}t_{4j-1}\frac{\d}{\d t_{4i+4j-2}}-D^{4i-1}, \Gamma_0]\\
&=&\frac12(4i-1) D^{4k-3}-\frac12\d^{2i-2}Q+2\sum_{j\geq 1}(i-j)t_{4j-1} \d^{2i+2j-2}\\ \notag
&&-\sum_{j=1}^{\infty}(2i+2j-1)t_{4j-1}D^{4i+4j-4}-[D^{4i-1},x-\frac12\sum_{k\geq 1}(4k-1)t_{4k-1} D^{4k-3}].
\]
Because
\[\notag [D^{4i-1},x]=[\d^{2i-1}D,x]=(2i-1)D^{4i-3}+\d^{2i-1}\theta\d_x=\frac12(4i-1) D^{4k-3}-\frac12\d^{2i-2}Q,\]

\[
[D^{4i-1},\frac12\sum_{k\geq 1}(4k-1)t_{4k-1} D^{4k-3}]=-\sum_{j\geq 1}(4j-1)t_{4j-1} D^{4i+4j-4},
\]
then
\[
[D_{4i-1}-D^{4i-1}, \Gamma_0]&=&0.
\]

For $ \Gamma_1,$ we have

\[
[D_{4i-2}-D^{4i-2}, \Gamma_1]&=&0,
\]

\[\notag
[D_{4i-1}-D^{4i-1}, \Gamma_1]&=&[\frac{\d}{\d t_{4i-1}}-\sum_{j=1}^{\infty}t_{4j-1}\frac{\d}{\d t_{4i+4j-2}}-D^{4i-1}, \theta+\sum_{k\geq 1}t_{4k-1} \d^{2k-1}]\\
&=&\d^{2i-1}-D^{4i-2}=0.
\]
For $ Q$, the following identities hold

\[
[D_{4i-2}-D^{4i-2}, Q]&=&0,
\]

\[
[D_{4i-1}-D^{4i-1}, Q]&=&[\frac{\d}{\d t_{4i-1}}-\sum_{j=1}^{\infty}t_{4j-1}\frac{\d}{\d t_{4i+4j-2}}-D^{4i-1}, Q]=0,
\]

\[\notag
[Q, \Gamma_0]&=&[Q, x]-[Q, \frac12\sum_{k\geq 1}t_{4k-1} \d^{2k-2}Q],\\
&=&-\theta-\sum_{k\geq 1}t_{4k-1} \d^{2k-1}=-\Gamma_1,
\]
\[
[Q, \Gamma_1]=[Q, \theta]=1.
\]
\end{proof}

Then it is easy to get  the following lemma by dressing structures.
\begin{lem}\label{thm-Mw}
The operators $M_j,\Q,L$ satisfy
\[
[\Q, M_0]&=&-M_1,\ \ \ [\Q, M_1]=1,\ \ \ [L^2, M_0]=1,
\]
\begin{equation}\label{bkpMt}
D_k M_j=[(L^k)_+,M_j],\ \ \ D_k \Q=[(L^k)_+,\Q],\quad\ k=4i-2,4i-1,\ i \in\Z_+.
\end{equation}

\end{lem}

\begin{proof}
The dressing structure
\[
\Phi[D_{4i-1}-D^{4i-1}, \Gamma_1]\Phi^{-1}=0;\]
will lead to
\[
[\Phi D_{4i-1}\Phi^{-1}-\Phi D^{4i-1}\Phi^{-1}, M_1]=0;\]
and further to
\[
[D_{4i-1}-\Phi_{4i-1}\Phi^{-1}+\sum_{j=0}^{\infty}t_{4j-1}\Phi_{4i+4j-2}\Phi^{-1}-L^{4i-1}, M_1]=0.\]
Then we get
\[[D_{4i-1}-(D_{4i-1}\Phi)\Phi^{-1}-L^{4i-1}, M_1]=0;\]
and using eq.\eqref{satoeq} we can derive
\[[D_{4i-1}-(L^{4i-1})_+, M_1]=0.\]
The other identities can be proved using the similar dressing techniques.
\end{proof}

From now on, we will introduce  the following operator $B_{mklp}$ defined as
\begin{align}\label{defBoperator}
B_{mklp}=M_0^kM_1^l\Q^pL^{2m}-(-1)^{pl+m+p+l}L^{2m-1}(\Q^p)M_1^{l}M_0^{k}L,\end{align}
where $k,m \geq 0; l , p = 0, 1.$
This operator is the generator of the additional symmetry of the SBKP hierarchy which shows  the difference between generators of the  SKP and SBKP hierarchies. For the SKP case in \cite{maninsuperKP}, in the construction of $B_{mklp}$, it contains only one term.

Then the following proposition can be got.
\begin{prp}The operator $B_{mnlp}$ satisfies the following flow equations
\begin{align}\label{Bflow}
& D_{4k-2}  B_{mnlp}=-[(L^{4k-2})_-,  B_{mnlp}],\ \  D_{4k-1} B_{mnlp}=-[(L^{4k-1})_-,  B_{mnlp}].
\end{align}
\end{prp}
\begin{proof}
The lemma can be proved by dressing the following identities by $\Phi$
\begin{align}
& [D_{4k-2}-D^{4k-2},  \Gamma_0^n\Gamma_1^lQ^p\d^{m}]=[D_{4k-1}-D^{4k-1}, \Gamma_0^n\Gamma_1^lQ^p\d^{m}]=0.
\end{align}
\end{proof}
To prove that $B_{mnlp}$  satisfies the B type condition, we need the following lemma.
\begin{lem}\label{BtypM}
The operators $M_i$ satisfy the following conjugate identities,
\eqa
M_i^*
=(-1)^iDL^{-1}M_iL D^{-1},\ \  \Q^*
=-DL^{-1}\Q L D^{-1}.\eeqa
\end{lem}
\begin{prf}
Using
\[
 \Phi^*=D\Phi^{-1} D^{-1},\ \Gamma_i^*
=(-1)^i\Gamma_i,\ \ Q^*=-Q,\]
  the following calculations
\[\notag
M_i^* =\Phi^{*-1}\Gamma_i^*\Phi^*=(-1)^iD\Phi D^{-1}\Gamma_i D\Phi^{-1} D^{-1}
=(-1)^iD\Phi D^{-1}\Phi^{-1}M_i\Phi D\Phi^{-1} D^{-1},\]

will lead to this lemma. The anti adjoint property of $\Q$ can be proved in a similar way.
\end{prf}

It is easy to check the following proposition holds basing on the Lemma \ref{BtypM} above.
\begin{prp}\label{asym}
The operator $B_{mklp}$ satisfies a  B type condition, namely
\begin{equation}
B_{mklp}^*=-D  B_{mklp} D^{-1}.
\end{equation}
\end{prp}
\begin{prf}
Using the Proposition \ref{BtypM}, the following calculation will lead to  this proposition

\begin{eqnarray*}B_{mklp}^*
&=&(M_0^kM_1^l\Q^pL^{2m}-(-1)^{pl+m+p+l}L^{2m-1}\Q^pM_1^{l}M_0^{k}L)^*\\
&=&(-1)^{pl}L^{2m*}(\Q^p)^*M_1^{l*}M_0^{k*}+(-1)^{m+p+l}L^{*}M_0^{k*}M_1^{l*}(\Q^p)^*L^{2m-1*}\\
&=&(-1)^{pl+m+p+l}DL^{2m-1}\Q^pM_1^{l}M_0^{k}LD^{-1}-DM_0^kM_1^l\Q^pL^{2m}D^{-1}\\
&=&-D(M_0^kM_1^l\Q^pL^{2m}-(-1)^{pl+m+p+l}L^{2m-1}\Q^pM_1^{l}M_0^{k}L)D^{-1}.
\end{eqnarray*}
\end{prf}

Basing on above proposition,  it is reasonable to define additional flows of the supersymmetric BKP hierarchy as
\begin{align}\label{addiflow}
&D_{mklp} L=[-({B}_{mklp})_-, L],\ \ k ,m\geq 0; l, p = 0, 1.
\end{align}
\begin{prp}\label{thm-st}
The flows \eqref{addiflow} commute with the flows of  the  supersymmetric BKP hierarchy. Namely, one has
\begin{equation}\label{st}
\left[D_{mnlp}, D_{k}\right]=0, \quad \ \ m, n\geq 0; l, p = 0, 1, ~~k=4i-2,4i-1,\ i \in\Z_+,
\end{equation}
which holds in the sense of acting on  $\Phi$.
\end{prp}

\begin{prf}
The  proposition can be checked case by case with the help of
eq.\eqref{Bflow}.
For example,

\begin{align}\label{}
&\left[D_{mnlp}, D_k\right]\Phi \nn\\
&=D_{mnlp}
 D_k\Phi- (-1)^{(l+p)k} D_kD_{mnlp}\Phi\nn\\
&=(-1)^{(l+p)k} [ ( L^k)_-,({B}_{mnlp})_-]\Phi
+[({B}_{mnlp})_-,  L^k]_- \Phi +(-1)^{(l+p)k}[(L^k)_+,{B}_{mnlp}]_-\Phi\\
& =0.\nn
\end{align}
\end{prf}

This proposition tells us that the additional flows of the supersymmetric BKP hierarchy are in fact its symmetries whose algebraic structure can be shown in the following proposition.
\begin{prp}
The algebra of additional symmetries of the SBKP hierarchy given
by eq.\eqref{addiflow} is isomorphic to the Lie algebra  $SW_{1+\infty}$.
\end{prp}
\begin{proof}
The isomorphism is given by
\[z&\mapsto\partial, \ \ \ \ \ \xi&\mapsto  Q+\Gamma_1\partial ,\\
\d_z&\mapsto \Gamma_0, \ \ \ \ \  \d_{\xi}&\mapsto  \Gamma_1,
\]
which further lead to
\[z&\mapsto L^2, \ \ \ \ \ \xi&\mapsto  \Q+M_1L^2 ,\\
\d_z&\mapsto M_0, \ \ \ \ \  \d_{\xi}&\mapsto  M_1.
\]
One can find the above construction keeps $\xi$ commuting with $z$.
\end{proof}

\section{Ghost symmetry of supersymmetric BKP hierarchy}
In this section, we will give another special symmetry which does not contain time variables explicitly. Before that, we firstly need to define the super-Baker-Akhiezer function (super-BA) and adjoint super-BA function as

\[\Phi_{BA}=\Phi e^{\xi},\ \ \ \Phi_{BA}^*=\Phi^{*-1} e^{-\xi},\]
where
\[\xi(\lambda,\eta,\theta,t)=\sum_{k=1}^{\infty}\lambda^{4k-2}t_{4k-2}+\eta\theta+(\eta-\lambda\theta)\sum_{k=1}^{\infty}\lambda^{4k-2}t_{4k-1},\ t_2\equiv x.\]
The following property can  be found
\[D_{4i-2}e^{\xi}=\d^{2i-1}e^{\xi},\ \ \ D_{4i-1}e^{\xi}=D^{2i-1}e^{\xi}.\]
Then we can prove that

\[L^2\Phi_{BA}=\lambda\Phi_{BA},\ \ \ L^{2*}\Phi_{BA}^*=-\lambda\Phi_{BA}^*,\]
\[D_{4i-2}\Phi_{BA}=(L^{4i-2})_+\Phi_{BA},\ \ \ D_{4i-2}\Phi_{BA}^*=-(L^{4i-2})_+^*\Phi_{BA}^*,\]
\[D_{4i-1}\Phi_{BA}=(L^{4i-1})_+\Phi_{BA},\ \ \ D_{4i-1}\Phi_{BA}^*=-(L^{4i-1})_+^*\Phi_{BA}^*.\]

The B type condition implies that the adjoint super-BA function $\Psi_{BA}$ can be in fact the supersymmetric derivative of its corresponding super-BA function $\Phi_{BA}$,
i.e.
\[\Psi_{BA}(t,\lambda)=-\lambda^{-1}\Phi_{BA}^{[1]}(t,-\lambda).\]
Then we can also define super-eigenfunctions of the supersymmetric BKP hierarchy  as
\[D_{4i-2}\phi=(L^{4i-2})_+\phi,\ \ \ D_{4i-2}\psi=-(L^{4i-2})_+^*\psi,\]
\[D_{4i-1}\phi=(L^{4i-1})_+\phi,\ \ \ D_{4i-1}\psi=-(L^{4i-1})_+^*\psi.\]

The super-eigenfunctions have the following spectral representation in term of integrals of Baker-Akhiezer functions as
\[\notag\phi(t,\theta)=\int d\lambda d\eta \phi(\lambda,\eta)\Phi_{BA}(t,\theta;\lambda,\eta),\ \ \psi(t,\theta)=\int d\lambda d\eta \bar \psi(\lambda,\eta)\Phi_{BA}^*(t,\theta;\lambda,\eta).\]

The B type condition also implies that the adjoint eigenfunction $\psi$ can be chosen as the supersymmetric derivative of its corresponding eigenfunction $\phi$,
i.e.
\[\psi=\phi^{[1]}.\]

The supersymmetric tau function of the supersymmetric BKP hierarchy can be defined by the residue (the coefficient before $D^{-1}$) of  supersymmetric Lax operators as
\[D_{4k-2}D\ln \tau=\res L^{4k-2},\ \ D_{4k-1}D\ln \tau=\res L^{4k-1}.\]

Define two eigenfunctions $\phi_1,\phi_2$ and the following operator
 \[B_g=\phi_1 D^{-1}\phi_2^{[1]}-(-1)^{| \phi_2|}\phi_2 D^{-1}\phi_1^{[1]},\ \ | \phi_1|\equiv | \phi_2|,\]
 which is used to generate the ghost flows.
According to
\[D^{-1}\phi=(-1)^{|\phi|} [\phi D^{-1}-D^{-1}\phi^{[1]}D^{-1}],\]
one can find the operator $B_g$ satisfies the B type condition,
i.e.
\[B_g^*=-DB_gD^{-1}.\]
Then the ghost flow of the supersymmetric BKP hierarchy can be defined as following
\[\label{ghostflow}
D_Z L=[B_g,L]=[\phi_1 D^{-1}\phi_2^{[1]}-(-1)^{| \phi_2|}\phi_2 D^{-1}\phi_1^{[1]}, L],\]
where functions $\phi_1,\phi_2$  are the eigenfunction and adjoint eigenfunction  of  the supersymmetric BKP hierarchy.
The following proposition will tell you the above flow is a symmetry of the supersymmetric BKP hierarchy.
\begin{prp}\label{symmetry}
The additional flow $D_Z$ commutes
with the supersymmetric BKP flows $D_{ n}$, i.e.,
\begin{eqnarray}
[D_Z, D_{n}]L=0, \ \ n=4i-2,4i-1,\ i \in\Z_+.
\end{eqnarray}
\end{prp}
\begin{proof}
Note that the derivatives $\phi_1^{[1]},\phi_2^{[1]}$ are in fact adjoint supersymmetric eigenfunctions.
The commutativity between ghost flows and supersymmetric BKP flows is in fact equivalent to the following Zero-Curvature equation which includes the following detailed proof
\begin{eqnarray*}&&D_ZB_n-D_{ n}( B_g)+[B_n, B_g]\\
&=&[B_g, L^n]_+- \phi_{1t_n} D^{-1}\phi_2^{[1]}- \phi_1 D^{-1}\psi^{[1]}_{2t_n}+(-1)^{| \phi_2|} [\phi_{2t_n} D^{-1}\phi_1^{[1]}+\phi_2 D^{-1}\phi^{[1]}_{1t_n}]\\
&&+[B_n, \phi_1 D^{-1}\phi_2^{[1]}-(-1)^{| \phi_2|}\phi_2 D^{-1}\phi_1^{[1]}]\\
&=&( B_n \phi_1 D^{-1}\phi_2^{[1]})_--(\phi_1 D^{-1}\phi_2^{[1]} B_n)_--P_0(B_n\phi_1) D^{-1}\phi_2^{[1]}+\phi_1 D^{-1}P_0(B^*_n\phi_2^{[1]})\\
&&-(-1)^{| \phi_2|}[( B_n\phi_2 D^{-1}\phi_1^{[1]} )_--(\phi_2 D^{-1}\phi_1^{[1]} B_n)_--P_0(B_n\phi_2) D^{-1}\phi_1^{[1]}+\phi_2 D^{-1}P_0(B^*_n\phi_1^{[1]})]\\
&=&0.\end{eqnarray*}
In the above proof the $P_0(A)$ means the coefficient over the term $D^0$ of the operator $A$.
\end{proof}

\section{The supersymmetric two-component BKP hierarchy}

Let us firstly define the supersymmetric two-component BKP hierarchy by two Lax operators.
Now we introduce the even and odd time variables $(t_2,t_3,t_6,\cdot;\hat t_2,\hat t_3,\hat t_6,\cdot)$ and the following definition of even and odd flows
\eqa
D_{4i-2}=\frac{\d}{\d t_{4i-2}},\ \ D_{4i-1}=\frac{\d}{\d t_{4i-1}}+\sum_{j=1}^{\infty}t_{4j-1}\frac{\d}{\d t_{4i+4j-2}},
\eeqa
\eqa
\hat D_{4i-2}=\frac{\d}{\d \hat t_{4i-2}},\ \ \hat D_{4i-1}=\frac{\d}{\d\hat  t_{4i-1}}+\sum_{j=1}^{\infty}\hat t_{4j-1}\frac{\d}{\d \hat t_{4i+4j-2}}.
\eeqa

This two families of
odd and even flows satisfy a nonabelian Lie superalgebra whose
commutation relations are
$$
[D_{4i-2},D_{4j-2}]=0~,\quad [D_{4i-2},D_{4j-1}]=0~,\quad
[D_{4i-1},D_{4j-1}]=-2 D_{4i+4j-2}~,
$$
\eqa[D_{4i-2},D]=0~,\quad [D_{4i-1},D]=0,\eeqa
$$
[\hat D_{4i-2},\hat D_{4j-2}]=0~,\quad [\hat D_{4i-2},\hat D_{4j-1}]=0~,\quad
[\hat D_{4i-1},\hat D_{4j-1}]=-2 \hat D_{4i+4j-2}~,
$$
\eqa[\hat D_{4i-2},D]=0~,\quad [\hat D_{4i-1},D]=0,\ \ [\hat D_{m},D_n]=0.\eeqa

The two Lax operators of the supersymmetric two-component BKP hierarchy will be defined in  forms as
\begin{equation} \label{PhP}
L= D+\sum_{i\ge1}u_i  D^{1-i}, \quad \hat{L}=
D^{-1}\hat{u}_{0}+\sum_{i\ge1}\hat{u}_i D^{i-1},\ | u_i|=i,|{\hat {u_i}}|=i+1,
\end{equation}
 such that
 \eqa \label{Bcondition}L^*=- D L D^{-1},\ \  \ \hat{L}^*=- D\hat{L}
D^{-1}.\eeqa
We call eqs.\eqref{Bcondition} the B type condition of the supersymmetric two-component BKP hierarchy.
The supersymmetric two-component BKP hierarchy is defined by the following
Lax equations:
\begin{align}\label{2bkpLax}
& D_{i}L=-[(L^{i})_-,  L], \quad  \quad \hat D_{i} \hat {L}=[(\hat{L}^{i})_+, \hat {L}],
\end{align}
\begin{align}\label{2bkpLax2}
& D_{i}\hat L=[(L^{i})_+, \hat L], \quad  \quad \hat D_{i} L=-[(\hat{L}^{i})_-,  {L}],\ \ i=4k-1, 4k-2,\ k\in\Z_+,
\end{align}
which is equivalent to the following equations
\begin{align}\label{2bkpLax3}
& D_{4k-2}L=[(L^{4k-2})_+,  L], \quad  \quad \hat D_{4k-2}\hat {L}=-[(\hat{L}^{4k-2})_-, \hat {L}],\\
& D_{4k-1}L=[(L^{4k-1})_+,  L]-2L^{4k}, \quad  \quad \hat D_{4k-1}\hat {L}=-[(\hat{L}^{4k-1})_-, \hat {L}]+2\hat L^{4k},\\
& D_{4k-2} \hat L=[(L^{4k-2})_+, \hat L], \quad  \quad \hat D_{4k-2}L=-[(\hat{L}^{4k-2})_-,  {L}],\\
& D_{4k-1} \hat L=[(L^{4k-1})_+, \hat L], \quad  \quad \hat D_{4k-1}L=-[(\hat{L}^{4k-1})_-,  {L}],
\end{align}
with $ k\in\Z_+$.

One can write the operators $L$ and $\hat{L}$ in a dressing form as
\begin{equation} \label{PPh}
L=\Phi D\Phi^{-1},\quad \hat{L}=\hat{\Phi} D^{-1}\hat{\Phi}^{-1},
\end{equation}
where
\begin{align} \label{dreop}
\Phi=1+\sum_{i\ge 1}a_i D^{-i},\quad \hat{\Phi}=1+\sum_{i\ge 1}b_i
D^{i},
\end{align}
 satisfy
\begin{equation}\label{phipsi2}
\Phi^*= D\Phi^{-1} D^{-1},\quad \hat{\Phi}^*= D\hat{\Phi}^{-1}
D^{-1}.
\end{equation}
Given $L$ and $\hat{L}$, the dressing operators $\Phi$ and
$\hat{\Phi}$ are determined uniquely up to a multiplication to the
right by operators with
constant coefficients. The dressing operators $\Phi$ and
$\hat{\Phi}$ take values in two separated B type Volterra groups. The supersymmetric two-component BKP hierarchy
can also be redefined as
\begin{align}
& D_{4k-2} \Phi=- (L^{4k-2})_-\Phi, \quad
\  D_{4k-2}\hat{\Phi}=(L^{4k-2})_+\hat{\Phi}, \label{ppt1}\\
&\hat D_{4k-2} \Phi=- (\hat{L}^{4k-2})_-\Phi, \quad \hat D_{4k-2}
\hat{\Phi}=(\hat{L}^{4k-2})_+\hat{\Phi}, \label{ppt2}
\end{align}
\begin{align}
&\hat D_{4k-1} \Phi=- (L^{4k-1})_-\Phi, \quad
 D_{4k-1}\hat{\Phi}=(L^{4k-1})_+ \hat{\Phi}, \label{ppt1}\\
&\hat D_{4k-1} \Phi=- (\hat{L}^{4k-1})_-\Phi, \quad \hat D_{4k-1}
\hat{\Phi}=(\hat{L}^{4k-1})_+\hat{\Phi}, \label{ppt2}
\end{align}
with $k\in\Z_+$.

Denote $t=(t_2,t_3,t_6,t_7,\dots)$,
$\hat{t}=(\hat{t}_2,\hat{t}_3,\hat{t}_6,\hat{t}_7,\dots)$ and introduce
two wave functions
\begin{align}\label{wavef}
w(z)=w(t, \hat{t}; z)=\Phi e^{\xi(t;z)},
\quad \hat{w}(z)=\hat{w}(t, \hat{t}; z)=\hat{\Phi}
e^{x z+\hat \xi(\hat{t};-z^{-1})},
\end{align}
where the functions $\xi,\hat \xi$ are defined as $\xi(t;
z)=\sum_{k\in\Z_+} t_{4k-2} z^{4k-2}+t_{4k-1} z^{4k-1},\ \hat \xi(\hat t;
z)=\sum_{k\in\Z_+} \hat t_{4k-2} \hat z^{4k-2}+\hat t_{4k-1} \hat z^{4k-1}$. It is easy to see
$D^i e^{x z}=z^i e^{x z},\ \ i\in\Z$
and
\[
L\,w(z)=z w(z), \quad \hat{L} \hat{w}(z) = z^{-1} \hat{w}(z).
\]

With the above preparation, it is time to  construct additional symmetries for the supersymmetric two-component BKP hierarchy in the next section.

\section{Additional symmetries of the supersymmetric two-component BKP hierarchy}

In this section, we are to construct additional symmetries for the supersymmetric two-component BKP hierarchy by using the Orlov--Schulman operators whose coefficients
depend explicitly on the time variables of the
hierarchy.

With the same dressing operators given in the eq.\eqref{dreop},
Orlov--Schulman operators $M_i,\hat M_i, \Q,\hat \Q$ are constructed in the following
dressing structure
\begin{equation*}\label{}
M_i=\Phi\Gm_i\Phi^{-1}, \quad
\hat{M}_i=\hat{\Phi}\hat{\Gm}_i\hat{\Phi}^{-1},\ \ i=0,1;\ \  \Q=\Phi Q\Phi^{-1},\hat \Q=\hat \Phi Q\hat \Phi^{-1},
\end{equation*}
where
\eqa\notag
\Gm_0 &=&x+\frac12\sum_{k\geq 1}(4k-2)t_{4k-2} D^{4k-4}+\frac12(4k-1)t_{4k-1} D^{4k-3}\\
&&-\frac12\sum_{k\geq 1}t_{4k-1} \d^{2k-2}Q+\sum_{i,j\geq 1}(i-j)t_{4i-1}t_{4j-1} \d^{2i+2j-2},\\ \notag
 \hat{\Gm}_0&=& x+\frac12\sum_{k\geq 1}(4k-2)\hat t_{4k-2}D^{-4k}+\frac12(4k-1)\hat t_{4k-1} D^{-1-4k}\\
&&-\frac12\sum_{k\geq 1}\hat t_{4k-1} \d^{-1-2k}Q+\sum_{i,j\geq 1}(i-j)\hat t_{4i-1}\hat t_{4j-1}\d^{-2i-2j},
\eeqa
\eqa
\Gm_1 &=&\theta+\sum_{k\in\Z_+}t_{4k-1} \d^{2k-1},\\
 \hat{\Gm}_1&=&\theta+\sum_{k\in\Z_+}\hat t_{4k-1} \d^{-2k},
\eeqa
where $Q=\d_{\theta}-\theta\d.$

Then one can derive the following lemma similarly as the case of the single-component supersymmetric BKP hierarchy.

\begin{lem}
The operators $\Gamma_i$ and $\hat{\Gamma_i}$ satisfy
\[
[D_{4i-2}-D^{4i-2}, \Gamma_i]&=&[D_{4i-1}-D^{4i-1}, \Gamma_i]=0,\]
\[
[\hat D_{4i-2}+D^{-4i+2}, \hat \Gamma_i]&=&[\hat D_{4i-1}+D^{-4i+1}, \hat \Gamma_i]=0;\]
\[
[D_{4i-2}, \hat \Gamma_i]&=&[D_{4i-1},\hat  \Gamma_i]=[\hat D_{4i-2},  \Gamma_i]=[\hat D_{4i-1},  \Gamma_i]=0;\]
where $i=0,1.$
\end{lem}

\begin{proof}

For the proof, we do the following calculation

\[
[\hat D_{4i-2}+D^{2-4i}, \hat\Gamma_0]&=&\frac12(4k-2) D^{-4k}+(1-2i)\d^{-2i}=0,
\]

\[&&
[\hat D_{4i-1}+D^{1-4i},\hat  \Gamma_0]\\
&=&[\frac{\d}{\d \hat t_{4i-1}}-\sum_{j=1}^{\infty}\hat t_{4j-1}\frac{\d}{\d \hat t_{4i+4j-2}}+D^{1-4i}, \hat \Gamma_0]\\
&=&\frac12(4i-1)D^{-1-4k}-\frac12\d^{-1-2i}Q+2\sum_{j\in\Z}(i-j)\hat t_{4j-1} \d^{-2i-2j}\\ \notag
&&-\sum_{j=1}^{\infty}(2i+2j-1)\hat t_{4j-1}D^{-4i-4j}+[D^{1-4i},x+\frac12\sum_{k\in\Z_+}(4k-1)\hat t_{4k-1} D^{-1-4k}].
\]
Because
\[\notag[D^{1-4i},x]=[\d^{-2i} D,x]=-2iD^{-4i-1}+\d^{-2i-1}\theta\d=-\frac12(4i-1)D^{-1-4i}+\frac12\d^{-1-2i}Q,\]

\[
[D^{1-4i},\frac12\sum_{k\in\Z_+}(4k-1)\hat t_{4k-1} D^{-1-4k}]=\sum_{j\in\Z}(4j-1)\hat t_{4j-1} D^{-4i-4j},
\]
then
\[
[\hat D_{4i-1}+D^{4i-1},\hat  \Gamma_0]&=&0.
\]

For $\hat  \Gamma_1,$ we get
\[\notag
[\hat D_{4i-1}+D^{1-4i},\hat  \Gamma_1]&=&[\frac{\d}{\d \hat t_{4i-1}}-\sum_{j=1}^{\infty}\hat t_{4j-1}\frac{\d}{\d\hat  t_{4i+4j-2}}+D^{1-4i}, \theta+\sum_{k\in\Z_+}\hat t_{4k-1}\d^{-2k}]\\
&=&\d^{-2i}-D^{-4i}=0.
\]
Further one can derive
\[
[\hat D_{k}, \Gamma_i]&=&
[D_{k}, \hat \Gamma_i]=0.
\]
For $ Q$, we can get

\[
[D_{k}, Q]&=&
[\hat D_{k}, Q]=0,
\]

\[
[Q, \hat \Gamma_0]&=&[Q,x]-[Q, \frac12\sum_{k\in\Z_+}\hat t_{4k-1} \d^{-1-2k}Q]\\
&=&-\theta-\sum_{k\in\Z_+}\hat t_{4k-1} \d^{-2k}=-\hat \Gamma_1.
\]
\end{proof}

Then it is easy to get  the following lemma using the above lemma and dressing structures.
\begin{lem}\label{thm-Mw}
The operators $M_j,\Q,L$,$\hat M_j,\hat \Q,\hat L$ satisfy
\[
[\Q, M_0]&=&-M_1,\ \ \ [\Q, M_1]=1,\ \ \ [L^2, M_0]=1,
\]
\begin{equation}\label{bkpMt}
D_k M_j=[(L^k)_+,M_j],\ \ \ D_k \Q=[(L^k)_+,\Q],\quad\ k=4i-2,4i-1,\ i \in\Z_+,
\end{equation}

\begin{equation}
\hat D_k  M_j=[-(\hat L^k)_-, M_j],\ \ \ \hat D_k  \Q=[-(\hat L^k)_-,\Q],\quad\ k=4i-2,4i-1,\ i \in\Z_+,
\end{equation}

\[
[\hat \Q, \hat M_0]&=&-\hat M_1,\ \ \ [\hat \Q,\hat  M_1]=1,\ \ \ [\hat L^{-2},\hat  M_0]=1,
\]

\begin{equation}
D_k\hat  M_j=[(L^k)_+,\hat M_j],\ \ \ D_k \hat \Q=[(L^k)_+,\hat \Q],\quad\ k=4i-2,4i-1,\ i \in\Z_+,
\end{equation}
\begin{equation}\label{bkpMt}
\hat D_k\hat  M_j=[-(\hat L^k)_-,\hat M_j],\ \ \ \hat D_k \hat \Q=[-(\hat L^k)_-,\hat \Q],\quad\ k=4i-2,4i-1,\ i \in\Z_+.
\end{equation}
\end{lem}

From now on, we will introduce  the following two operators $B_{mklp}$ and $\hat{B}_{mklp}$,
given any pair of integers $(m,k,l,p)$ with $m,k\ge0,l,p=0,1,$ as
\begin{align}\label{defBoperator}
B_{mklp}=M_0^kM_1^l\Q^pL^{2m}-(-1)^{pl+m+p+l}L^{2m-1}(\Q^p)M_1^{l}M_0^{k}L, \\
\ \ \hat{B}_{mklp}=\hat M_0^k\hat M_1^l\hat \Q^p\hat L^{-2m}-(-1)^{pl+m+p+l}\hat L^{-2m+1}(\hat \Q^p)\hat M_1^{l}\hat M_0^{k}\hat L^{-1}.
\end{align}
As a corollary, the following proposition can be got.
\begin{prp}
For any $\bar B_{mklp}=B_{mklp}, \hat{B}_{mklp}$, one has
\begin{align}\label{Bflow2}
& D_n \bar B_{mklp}=[(L^n)_+, \bar B_{mklp}], \quad \hat D_n\bar B_{mklp}=[-(\hat L^n)_-,\bar B_{mklp}],\ \ n=4i-2,4i-1,\ i \in\Z_+.
\end{align}
\end{prp}

To prove that $B_{mklp}$ and $\hat B_{mklp}$ satisfy the B type condition, we need the following lemma.
\begin{lem}\label{BtypM2}
Operators $M_i$and $\hat M_i$ satisfy the following conjugate identities,
\eqa
M_i^*
=(-1)^iDL^{-1}M_iL D^{-1},\ \ \ \hat M_i^*
=(-1)^iD\hat L\hat M_i\hat L^{-1} D^{-1},\eeqa
\[
  \Q^*
=-DL^{-1}\Q L D^{-1},\ \ \hat \Q^*
=-D\hat L\hat \Q \hat L^{-1} D^{-1}.\]
\end{lem}
\begin{prf}

Using
\[
 \Phi^*=D\Phi^{-1} D^{-1},\ \ \hat\Phi^*=D\hat\Phi^{-1} D^{-1},\]
  the following calculations
\[\notag
M_i^* =\Phi^{*-1}\Gamma_i^*\Phi^*=(-1)^iD\Phi D^{-1}\Gamma_i D\Phi^{-1} D^{-1}
=(-1)^iD\Phi D^{-1}\Phi^{-1}M_i\Phi D\Phi^{-1} D^{-1},\]
\[\notag
\hat M_i^* =\hat \Phi^{*-1}\hat \Gamma_i^*\hat \Phi^*=(-1)^iD\hat \Phi D^{-1}\hat \Gamma_i D\hat \Phi^{-1} D^{-1}
=(-1)^iD\hat \Phi D^{-1}\hat \Phi^{-1}\hat M_i\hat \Phi D\hat \Phi^{-1} D^{-1},\]
will lead to this lemma.
\end{prf}

It is easy to check the following proposition holds basing on the Lemma \ref{BtypM2} above.
\begin{prp}\label{asym}
Operators $B_{mklp}$ and $\hat B_{mklp}$ satisfy the B type condition, namely
\begin{equation}
B_{mklp}^*=-D  B_{mklp} D^{-1}, \quad \hat{B}_{mklp}^*=-D \hat{B}_{mklp}
 D^{-1}.
\end{equation}
\end{prp}
\begin{prf}
Using the Proposition \ref{BtypM2}, the following calculation will lead to the first identity of this proposition
\begin{eqnarray*}\hat B_{mklp}^*
&=&(\hat M_0^k\hat M_1^l\hat \Q^p\hat L^{2m}-(-1)^{pl+m+p+l}\hat L^{2m-1}\hat \Q^p\hat M_1^{l}\hat M_0^{k}\hat L)^*\\
&=&(-1)^{pl}\hat L^{2m*}(\hat \Q^p)^*\hat M_1^{l*}\hat M_0^{k*}+(-1)^{m+p+l}\hat L^{*}\hat M_0^{k*}\hat M_1^{l*}(\hat \Q^p)^*\hat L^{2m-1*}\\
&=&(-1)^{pl+m+p+l}D\hat L^{2m-1}\hat \Q^p\hat M_1^{l}\hat M_0^{k}\hat LD^{-1}-D\hat M_0^k\hat M_1^l\hat \Q^p\hat L^{2m}D^{-1}\\
&=&-D(\hat M_0^k\hat M_1^l\hat \Q^p\hat L^{2m}-(-1)^{pl+m+p+l}\hat L^{2m-1}\hat \Q^p\hat M_1^{l}\hat M_0^{k}\hat L)D^{-1}.
\end{eqnarray*}
 The second identity can be proved in a similar way.
\end{prf}

Now we can define  the following additional equations as
\begin{align}
&D_{mklp}\Phi=- (B_{mklp})_-\Phi,
\quad D_{mklp}\hat\Phi=(B_{mklp})_+\hat{\Phi}, \label{add}\\
&\hat D_{mklp}\Phi=- (\hat{B}_{mklp})_-\Phi, \quad
 \hat D_{mklp}\hat\Phi=(\hat{B}_{mklp})_+\hat{\Phi}. \label{add2}
\end{align}
These equations are equivalent to the following Lax equations
\begin{align}
&D_{mklp} L=[- (B_{mklp})_-,L],
\quad D_{mklp}\hat{L}=[(B_{mklp})_+,\hat L], \label{add}\\
&\hat D_{mklp} L=[- (\hat{B}_{mklp})_-,L], \quad
\hat D_{mklp}\hat{L}=[(\hat{B}_{mklp})_+,\hat{L}]. \label{add2}
\end{align}
Similarly, we will get the following proposition.
\begin{prp}\label{thm-st}
The flows \eqref{add} and \eqref{add2} commute with the flows of  the  supersymmetric two-component BKP hierarchy. Namely, for any $\bar D_{mnlp}=D_{mnlp}, \hat D_{mnlp}$
and $\bar D_k=D_k, \hat D_k$ one has
\begin{equation}\label{st}
\left[ \bar D_{mnlp},\bar D_k\right]=0, \quad m,n\in\Z_+; l,p=0,1; ~~ k=4i-2,4i-1,\ i \in\Z_+,
\end{equation}
which holds in the sense of acting on  $\Phi$ or $\hat\Phi$.
\end{prp}
\begin{prf}
The  proposition can be checked case by case with the help of
eq.\eqref{Bflow2} and eqs.\eqref{add}-\eqref{add2}.
For example,

\begin{align}\label{}
&\left[D_{mnlp},
\hat D_k\right]\Phi \nn\\
&=D_{mnlp}
\hat D_k\Phi- (-1)^{(l+p)k}\hat D_kD_{mnlp}\Phi\nn\\
=&(-1)^{(l+p)k} [ (\hat L^k)_-,({B}_{mnlp})_-]\Phi
-[({B}_{mnlp})_+, \hat L^k]_- \Phi -(-1)^{(l+p)k}[(\hat L^k)_-,{B}_{mnlp}]_-\Phi =0, \nn
\end{align}
\begin{align}\label{}
&\left[ \hat D_{mnlp}, D_k\right]\hat{\Phi} \nn\\
=& (-1)^{(l+p)k}[ (L^k)_+,(\hat{B}_{mnlp})_+]\hat{\Phi}
+[-(\hat{B}_{mnlp})_-, L^k]_+ \hat{\Phi}-(-1)^{(l+p)k}[(L^k)_+,\hat{B}_{mnlp}]_+\hat{\Phi} =0. \nn
\end{align}
The other cases can be proved in similar ways. This is the end of this proposition.
\end{prf}

Similarly as the SBKP hierarchy, the algebraic structure of the additional symmetry of the S2BKP hierarchy will be talked about in the next proposition.
\begin{prp}
The algebra of additional symmetries of the two-component SBKP hierarchy  is isomorphic to the Lie algebra of super quasi-differential operators, which is isomorphic (as a Lie
algebra) to $SW_{1+\infty}\bigoplus SW_{1+\infty} $.
\end{prp}
\begin{proof}
The isomorphism is given by
\[z&\mapsto\partial, \ \ \ \ \ \xi&\mapsto  Q+\Gamma_1\partial ,\\
\d_z&\mapsto \Gamma_0, \ \ \ \ \  \d_{\xi}&\mapsto  \Gamma_1,
\]
\[\hat z&\mapsto\partial, \ \ \ \ \ \hat \xi&\mapsto Q+\hat \Gamma_1\partial ,\\
\d_{\hat z}&\mapsto \hat \Gamma_0, \ \ \ \ \  \d_{\hat \xi}&\mapsto  \hat \Gamma_1,
\]
which further lead to
\[z&\mapsto L^2, \ \ \ \ \ \xi&\mapsto  \Q+M_1L^2 ,\\
\d_z&\mapsto M_0, \ \ \ \ \  \d_{\xi}&\mapsto  M_1,
\]
\[\hat z&\mapsto\hat  L^{-2}, \ \ \ \ \ \hat \xi&\mapsto \hat  \Q+\hat M_1\hat L^{-2} ,\\
\d_{\hat z}&\mapsto \hat M_0, \ \ \ \ \  \d_{\hat \xi}&\mapsto  \hat M_1.
\]
One can find the above construction keeps $\xi$ commuting with $z$ and $\hat\xi$ commuting with $\hat z$.
\end{proof}

 If we do a (4n,2)-reduction from the supersymmetric two-component BKP hierarchy,
 a reduced hierarchy called the supersymmetric D type Drinfeld--Sokolov
hierarchies with a supersymmetric Block type additional  symmetry which will be discussed in the next section.

\section{Supersymmetric D type Drinfeld--Sokolov hierarchy}

Assume a new Lax operator $\cL$ which has the following relation with two Lax operators of the supersymmetric two-component BKP hierarchy introduced in the last section
\begin{equation}\label{constraint}\cL=L^{4n}=\hat{L}^{2},\ n\geq 2.\end{equation}
Then the Lax operators of the supersymmetric two-component BKP hierarchy will be reduced to the following Lax operator of the supersymmetric D type Drinfeld--Sokolov hierarchy whose Bosonic case can be seen in \cite{LWZ,bkpds,BKP-DS}
\begin{equation}\label{mL}
   \cL=D^{4n}+\sum_{i=1}^{n} D^{-1}\left(v_i
D^{4i-1}+D^{4i-1} v_i\right) +D^{-1} \rho D^{-1} \rho;\ |v_i|=0,|\rho|=1.
\end{equation}

One can easily find the Lax operator $\cL$ of the supersymmetric D type Drinfeld--Sokolov hierarchy will  satisfy the following B type condition
\eqa\label{symcL}\cL^*=D \cL D^{-1}.\eeqa
This Lax operator $\cL$ of the supersymmetric D type Drinfeld--Sokolov hierarchy has the following dressing structure\cite{bkpds}
\begin{equation} \label{dress}
\cL=\Phi D^{4n}\Phi^{-1}=\hat{\Phi} D^{-2}\hat{\Phi}^{-1}.
\end{equation}
Here
\begin{align} \label{Phi}
\Phi=1+\sum_{i\ge 1}a_i D^{-i},\quad \hat{\Phi}=1+\sum_{i\ge 1}b_i
D^{i}
\end{align}
are pseudo supersymmetric differential operators  that also
satisfy the following B type condition
\begin{equation}\label{phipsi3}
\Phi^*= D\Phi^{-1} D^{-1},\quad \hat{\Phi}^*= D\hat{\Phi}^{-1}
D^{-1}.
\end{equation}
The dressing structures inspire us to define two fractional operators as
\begin{equation}
\cL^{\frac1{4n}}= D+\sum_{i\ge1}u_i  D^{-i}, \quad \cL^{\frac12}=
D^{-1}\hat{u}_{-1}+\sum_{i\ge1}\hat{u}_i D^i.
\end{equation}

Two fractional operators $\cL^{\frac1{4n}}$ and $\cL^{\frac12}$ can be rewritten  in a dressing form as
\begin{equation} \label{PPh}
\cL^{\frac1{4n}}=\Phi D\Phi^{-1},\quad \cL^{\frac12}=\hat{\Phi} D^{-1}\hat{\Phi}^{-1}.
\end{equation}

 The   supersymmetric D type Drinfeld--Sokolov hierarchy being considered in this paper is defined by the following
Lax equations:
\begin{align}\label{PPht}
& D_k  \cL=[(\cL^{\frac k{4n}})_+, \cL], \quad \hat D_k  \cL=[-(\cL^{\frac k2})_-, \cL],\ \ k=4i-2,4i-1,\ i \in\Z_+.
\end{align}

The dressing operators $\Phi$ and $\hat \Phi$ are same as the ones of the supersymmetric two-component BKP hierarchy.
Given $\cL$, the dressing operators $\Phi$ and
$\hat{\Phi}$ are  uniquely determined up to a multiplication to the
right by operators of the form \eqref{Phi} and \eqref{phipsi3} with
constant coefficients. The supersymmetric D type Drinfeld-Sokolov hierarchies
can also be redefined as the following Sato equations
\begin{align}
&D_k \Phi=- (\cL^{\frac k{4n}})_-\Phi, \quad
D_k  \hat{\Phi}=(\cL^{\frac k{4n}})_+ \hat{\Phi}, \label{ppt1}\\
&\hat D_k  \Phi=- (\cL^{\frac k{2}})_-\Phi, \quad \hat D_k \hat{\Phi}=(\cL^{\frac k{2}})_+\hat{\Phi}, \label{ppt2}
\end{align}
with $k=4i-2,4i-1,\ i \in\Z_+$.

After the above preparation, we will show that this supersymmetric D type Drinfeld-Sokolov hierarchy has a nice Block symmetry as its appearance  in the Bigraded Toda hierarchy \cite{ourBlock}.

\section{Supersymmetric Block symmetries of supersymmetric D type Drinfeld-Sokolov hierarchies }

In this section, we will put the constrained condition
eq.\eqref{constraint} into the construction of the flows of the additional
symmetry which form a $N=2$ supersymmetric extension of the well-known Block algebra \cite{Block}.

With the dressing operators given in eq.\eqref{PPh}, we introduce two new Orlov-Schulman operators as following
\eqa \label{cMandM}\cM_i=M_iL^{2-4n},\ \ \hat \cM_i=\hat M_i\hat L^{-4}.\eeqa

It is easy to see the following lemma holds.
\begin{lem}\label{thm-Mw}
The operators $\cM_j$ and $\hat{\cM_j}$ satisfy
\begin{equation}\label{}
[\cL, \cM_0]=1, \quad [\cL,\hat{\cM_0}]=1;\ \
[\hat \Q, \hat \cM_0]=-\hat \cM_1;
\end{equation}

and
\begin{equation}\label{bkpMt}
D_k \cM_j=[(\cL^{\frac k{4n}})_+,\cM_j],\ \ \ D_k \Q=[(\cL^{\frac k{4n}})_+,\Q],\quad\ k=4i-2,4i-1,\ i \in\Z_+,
\end{equation}

\begin{equation}
\hat D_k  \cM_j=[-(\cL^{\frac k2})_-, \cM_j],\ \ \ \hat D_k  \Q=[-( \cL^{\frac k2})_-,\Q],\quad\ k=4i-2,4i-1,\ i \in\Z_+,
\end{equation}

\begin{equation}
D_k\hat  \cM_j=[(\cL^{\frac k{4n}})_+,\hat \cM_j],\ \ \ D_k \hat \Q=[(\cL^{\frac k{4n}})_+,\hat \Q],\quad\ k=4i-2,4i-1,\ i \in\Z_+,
\end{equation}
\begin{equation}\label{bkpMt}
\hat D_k\hat  \cM_j=[-(\cL^{\frac k2})_-,\hat \cM_j],\ \ \ \hat D_k \hat \Q=[-(\cL^{\frac k2})_-,\hat \Q],\quad\ k=4i-2,4i-1,\ i \in\Z_+,
\end{equation}
which can be simplified to
\begin{equation}\label{Mt}
D_k \bar\cM_j=[(\cL^{\frac k{4n}})_+,\bar\cM_j],\quad
 \hat D_k\bar\cM_j=[-(\cL^{\frac k{2}})_-, \bar\cM_j],
\end{equation}
where
 $\bar{\cM_j}=\cM_j$ or $\hat{\cM_j}, k=4i-2,4i-1,\ i \in\Z_+$.
\end{lem}

To make the operators used in the additional symmetry satisfying the B type condition, we need to prove the following $B$ type property of $\cM_i-\hat \cM_i$ which is
included in the following lemma.

\begin{lem}\label{asymM-M2}
The difference of two Orlov-Schulman operators $\cM_0$ and $\hat \cM_0$ for the supersymmetric D type Drinfeld-Sokolov hierarchy has the following D type property:
\begin{align}
\cL^*(\cM_0-\hat \cM_0)^*=-D\cL(\cM_0 -\hat\cM_0 )D^{-1}.
\end{align}
\end{lem}
\begin{prf}
It is easy to find  the two Orlov-Schulman operators $\cM_0$ and $\hat \cM_0$ of the supersymmetric D type Drinfeld-Sokolov hierarchy can be expressed by Orlov-Schulman operators $M_0,\hat M_0$ and Lax operators $L,\hat L$ of the supersymmetric two-component BKP hierarchy as

\eqa \label{cMandM}\cM_0=M_0L^{2-4n},\ \ \hat \cM_0=-\hat M_0\hat L^{-4}.\eeqa

Using Lemma \ref{BtypM2}, putting eq.\eqref{cMandM} into $(\cM_0-\hat \cM_0)^*$ can lead to
\begin{align}
&(\cM_0-\hat \cM_0)^*=-DL^{1-4n}M_0LD^{-1}+D\hat L^{-3}\hat M_0\hat L^{-1} D^{-1}\\
&=-D(L^{1-4n}M_0-L^{-4n})D^{-1}+D(\hat L^{-4}\hat M_0+\hat L^{-2}) D^{-1},
\end{align}
which can further lead to
\begin{align}
\cL^*(\cM_0-\hat \cM_0)^*=-D(\cL\cM_0 -\cL\hat\cM_0)D^{-1}.
\end{align}

In the above calculation, the commutativity between $\cL$ and $\cM_0 -\hat\cM_0$ is already used.
Till now, the proof is finished.
\end{prf}

One can also get
\begin{align}
&\cM_i^*=-DL^{-1}\cM_iLD^{-1}, \ \ \hat \cM_i^*=-D\hat L\hat \cM_i\hat L^{-1} D^{-1}.
\end{align}

For the supersymmetric D-type Drinfeld-Sokolov hierarchy, we define the additional operator $\B_{mk}^{lp\hat l\hat p}$ as
\begin{align}
\B_{mk}^{lp\hat l\hat p}=(\cM_0-\hat \cM_0)^{k}(M_1^{l}\Q^p\hat M_1^{\hat l}\hat \Q^{\hat p}-(-1)^{\prod+\sum+k}\hat L^{-1}\hat \Q^{p}\hat M_1^{\hat l}\hat LL\Q^{p}M_1^{l}L^{-1})\cL^m,\ \
\end{align}
where $l,p,\hat l,\hat p=0,1;\ \prod=\prod_{a,b=l,p,\hat l,\hat p}ab,\sum=\sum_{a=l,p,\hat l,\hat p}a.$

One can get the following proposition.
\begin{prp}\label{asym}
The operator $\B_{mk}^{lp\hat l\hat p}$ satisfies a B type condition, namely
\begin{align}
(\B_{mk}^{lp\hat l\hat p})^*=-D\B_{mk}^{lp\hat l\hat p}D^{-1},\ \  l,p,\hat l,\hat p=0,1.
\end{align}
\end{prp}
\begin{prf}
Using the Proposition \ref{BtypM2}, the following calculation will lead to
\begin{eqnarray*}(\B_{mk}^{lp\hat l\hat p})^*
&=&((\cM_0-\hat \cM_0)^{k}(M_1^{l}\Q^p\hat M_1^{\hat l}\hat \Q^{\hat p}-(-1)^{\prod+\sum+k}\hat L^{-1}\hat \Q^{p}\hat M_1^{\hat l}\hat LL\Q^{p}M_1^{l}L^{-1})\cL^m)^*\\
&=&\cL^{m*}[(-1)^{\prod}(\hat \Q^{\hat p})^*\hat M_1^{\hat l*}\Q^{p*}M_1^{l*}+(-1)^{\sum+k}L^{-1*}M_1^{l*}\Q^{p*} L^*\hat L^*\hat M_1^{\hat l*}\hat \Q^{\hat p*}\hat L^{-1*}](\cM_0-\hat \cM_0)^{k*}\\
&=&D\cL^{m}[(-1)^{\prod+\sum+k}\hat L^{-1}\hat \Q^{\hat p}\hat M_1^{\hat l}\hat LL\Q^{p}M_1^{l}L^{-1}-(M_1^{l}\Q^p\hat M_1^{\hat l}\hat \Q^{\hat p})](\cM_0-\hat \cM_0)^{k}D^{-1}\\
&=&-D(\cM_0-\hat \cM_0)^{k}(M_1^{l}\Q^p\hat M_1^{\hat l}\hat \Q^{\hat p}-(-1)^{\prod+\sum+k}\hat L^{-1}\hat \Q^p\hat M_1^{\hat l}\hat LL\Q^{p}M_1^{l}L^{-1})\cL^{m}D^{-1}.
\end{eqnarray*}
\end{prf}

That means it is reasonable to define the additional flow of the supersymmetric D type Drinfeld--Sokolov hierarchy as
\begin{align}\label{blockflow}
&\frac{\pd \cL}{\pd c_{mk}^{lp\hat l\hat p}}=[-({\B}_{mk}^{lp\hat l\hat p})_-, \cL],\ \ l,p,\hat l,\hat p=0,1; m,k\in \Z_+.
\end{align}
Whether these additional flows are symmetries of the supersymmetric D type Drinfeld--Sokolov hierarchy will be answered in the next proposition.

\begin{prp}
The flows in eq.\eqref{blockflow} can commute with original flows of  the supersymmetric Drinfeld--Sokolov hierarchy of type $D$, namely,
\begin{equation*}
\left[\frac{\pd}{\pd c_{mk}^{lp\hat l\hat p}}, D_n\right]=0, \quad
\left[\frac{\pd}{\pd c_{mk}^{lp\hat l\hat p}}, \hat{D}_n\right]=0,
\end{equation*}
where $\ l,p,\hat l,\hat p=0,1; m,k\in \Z_+,n=4i-2,4i-1,\ i \in\Z_+,$
which hold in the sense of acting on  $\Phi$, $\hat\Phi$ or $\cL.$

\end{prp}
\begin{prf} According to the definition,
\begin{eqnarray*}
[\partial_{c_{mk}^{lp\hat l\hat p}},D_n]\Phi=\partial_{c_{mk}^{lp\hat l\hat p}}
(D_n\Phi)-
D_n (\partial_{c_{mk}^{lp\hat l\hat p}}\Phi),
\end{eqnarray*}
and using the actions of the additional flows and the
flows of the D type  Drinfeld-Sokolov hierarchy on $\Phi$,  we have
\begin{eqnarray*}
[\partial_{c_{mk}^{lp\hat l\hat p}},D_n]\Phi
&=& -\partial_{c_{mk}^{lp\hat l\hat p}}\left((\cL^{\frac{k}{4n}})_{-}\Phi\right)+
D_n \left((\B_{mk}^{lp\hat l\hat p})_{-}\Phi \right)\\
&=& -(\partial_{c_{mk}^{lp\hat l\hat p}}\cL^{\frac{k}{4n}} )_{-}\Phi-
(\cL^{\frac{k}{4n}})_{-}(\partial_{c_{mk}^{lp\hat l\hat p}}\Phi)\\&&+
[D_n (\B_{mk}^{lp\hat l\hat p})]_{-}\Phi +
(\B_{mk}^{lp\hat l\hat p})_{-}(D_n\Phi).
\end{eqnarray*}
Using eq.\eqref{PPht} and eq.\eqref{Mt}, it
equals
\begin{eqnarray*}
[\partial_{c_{mk}^{lp\hat l\hat p}},D_n]\Phi
&=&[\left(\B_{mk}^{lp\hat l\hat p}\right)_{-}, \cL^{\frac{k}{4n}}]_{-}\Phi+
(\cL^{\frac{k}{4n}})_{-}\left(\B_{mk}^{lp\hat l\hat p}\right)_{-}\Phi\\
&&+[(\cL^{\frac{k}{4n}})_{+},\B_{mk}^{lp\hat l\hat p}]_{-}\Phi-(\B_{mk}^{lp\hat l\hat p})_{-}(\cL^{\frac{k}{4n}})_{-}\Phi\\
&=&[(\B_{mk}^{lp\hat l\hat p})_{-}, \cL^{\frac{k}{4n}}]_{-}\Phi- [\B_{mk}^{lp\hat l\hat p},
(\cL^{\frac{k}{4n}})_{+}]_{-}\Phi\\&&+
[(\cL^{\frac{k}{4n}})_{-},(\B_{mk}^{lp\hat l\hat p})_{-}]\Phi\\
&=&0.
\end{eqnarray*}
The other cases of this proposition can be proved in similar ways.
\end{prf}

The above proposition indicates that eq.\eqref{blockflow} is a symmetry of the supersymmetric D type Drinfeld-Sokolov hierarchy.
Further we can prove that the following identities hold true
\begin{equation}\label{blockflowM}
\frac{\pd \cM_i}{\pd c_{mk}^{lp\hat l\hat p}}=[-(\B_{mk}^{lp\hat l\hat p})_-,\cM_i], \quad
\frac{\pd\hat{\cM_i}}{\pd c_{mk}^{lp\hat l\hat p}}=[(\B_{mk}^{lp\hat l\hat p})_+,\hat{\cM}_i],
\end{equation}
\begin{equation}\label{waves}
\frac{\pd w(z^{\frac 1{4n}})}{\pd c_{mk}^{lp\hat l\hat p}}=-(\B_{mk}^{lp\hat l\hat p})_-w(z^{\frac 1{4n}}), \quad
\frac{\pd\hat{w}(z^{\frac 12})}{\pd c_{mk}^{lp\hat l\hat p}}=(\B_{mk}^{lp\hat l\hat p})_+\hat{w}(z^{\frac 12}).
\end{equation}

Using same techniques used in \cite{ourBlock}, the following theorem can be derived.
\begin{thm}\label{thm-MLs}
The flows in eq.\eqref{blockflow} about  additional symmetries of supersymmetric D type Drinfeld-Sokolov hierarchy compose a supersymmetric  type Block  Lie algebra which contains the following Block Lie algebra while $l=p=\hat l=\hat p=0$

\[\label{blocks}[\partial_{c_{ml}^{0000}},\partial_{c_{sk}^{0000}}]=(km-s l)\partial_{c_{m+s-1,k+l-1}^{0000}},\ \  m,s, k,l\in \Z_+,\]

which holds in the sense of acting on  $\Phi$, $\hat\Phi$ or $\cL.$
\end{thm}
\begin{prf}
The similar proof for eq.\eqref{blocks} can be found in our paper \cite{BKP-DS}.
\end{prf}

This is one kind of supersymmetric extensions of the Block algebra because it is involved with  the supersymmetric variables and supersymmetric derivative $D$. However, its algebra structure is still not clear now, which  deserves further study in the future.

  \section{Conclusions and Discussions}

Our earlier papers  show that the Block type algebras
appear not only in Toda type difference systems
  but also in  differential systems such as two-BKP hierarchy, D type
  Drinfeld-Sokolov hierarchy \cite{BKP-DS}. The above results show that in their corresponding supersymmetric systems, there also exists the hidden $N=2$ Block type supersymmetric
  algebraic structures. These results further show that the
  Block type infinite dimensional Lie algebra has a certain of
  universality in integrable hierarchies.

  Although the supersymmetric two-component BKP hierarchy and supersymmetric Drinfeld-Sokolov hierarchy of type D might  be not available in the fermionic string theory now comparing with the application of the classical KP hierarchy to the bosonic string, they have a great advantage of showing their superconformal structures. In this paper, we also show the structure of a super Block algebra of the supersymmetric two-BKP hierarchy and its reduced hierarchy.
  The superconformal algebra may appear  in the related Hamiltonian structures by considering the reductions of super-BKP hierarchies like reductions of the super-KP system to super-KdV system. This may be an interesting subject for our future study which may relate the supersymmetric BKP systems in this paper to problems in physics. There are also some other interesting subjects such as the relation of the hierarchies introduced in this paper with the  quantum spin chains  as in \cite{Spinchains}.
These directions might be included in our future study.
%%%%%%%%%%%%%%%%%%%%%%%%%%%%%%%%%%%%%%%%%%%%

\vskip 0.5truecm \noindent{\bf Acknowledgments.}
Chuanzhong Li is supported by the National Natural Science Foundation of China under Grant No. 11201251,
 the Zhejiang Provincial Natural Science Foundation under Grant No. LY15A010004, LY12A01007, the Natural Science Foundation of Ningbo under Grant No. 2013A610105, 2014A610029. Jingsong He is supported by the National Natural Science Foundation of China under Grant No. 11271210, K. C. Wong Magna Fund in
Ningbo University.

\end{document}